\documentclass{iopart}
\usepackage     [utf8x]                 {inputenc}
\usepackage     [T1]                    {fontenc}
\usepackage                             {color}
\usepackage                             {amsmath,amsthm, amssymb}
\usepackage                             {bbold}
\usepackage                             {graphicx}
\usepackage     [caption=false]         {subfig}
\usepackage     [english]               {babel}
\usepackage                             {cite}
\usepackage                             {hyperref}

\definecolor{darkblue}{rgb}{0.0,0.0,0.5}
\hypersetup{
    colorlinks,
    linkcolor=darkblue,
    citecolor=darkblue,
    urlcolor=darkblue,
    pdftitle={Bose-Hubbard model on two-dimensional line graphs}}

\newtheorem*{theorem}{Theorem}
\newtheorem{prop}{Proposition}
\newtheorem*{cor}{Corollary}
\newtheorem*{defin}{Definition}

\begin{document}

\title{Bose-Hubbard model on two-dimensional line graphs}
\author{Johannes Motruk, Andreas Mielke}
\address{Institut für Theoretische Physik, Ruprecht-Karls-Universität Heidelberg, Philosophenweg 19, D-69120 Heidelberg, Germany}
\ead{j.motruk@thphys.uni-heidelberg.de, mielke@tphys.uni-heidelberg.de}
\date{\today}

\begin{abstract}
We construct a basis for the many-particle ground states of the positive hopping Bose-Hubbard model on line graphs of finite 2-connected planar bipartite graphs at sufficiently
low filling factors. The particles in these states are localized on non-intersecting vertex-disjoint cycles of the line graph which correspond to non-intersecting edge-disjoint cycles
of the original graph. The construction works up to a critical filling factor at which the cycles are close-packed. 
\end{abstract}


\submitto{\JPA}

\maketitle

\section{Introduction}

The Hubbard model on lattices that are line graphs exhibits a flat band in the single-particle spectrum \cite{0305-4470-24-14-018}. We consider the case of positive hopping matrix elements, the
flat band is then the lowest band and gives rise to a large degeneracy of single-particle ground states. The space of these ground states is spanned by localized states. In this work
we investigate this model for bosons.

The problem is closely related to spin systems with antiferromagnetic exchange interactions on frustrated lattice
geometries \cite{0305-4470-39-34-006,PhysRevB.79.054403,PhysRevB.70.104415,PhysRevLett.88.167207,PhysRevB.70.100403,PTPS.160.361}. If these systems are under the influence of
a strong external magnetic field which nearly fully polarizes the spins, single spin-flips, so called magnons behave like hopping bosons on a lattice. The corresponding single-particle
basis consists of localized magnon states \cite{PhysRevLett.88.167207}.
Exact many-particle states may be constructed by placing these localized states on the lattice such that they do not
overlap. Schmidt \emph{et al.} investigated the linear independence of the many-particle states obtained in this way for several lattice geometries containing
two-dimensional line graphs like the kagome or checkerboard lattice \cite{0305-4470-39-34-006}. Their linear independence was proved but it became as well obvious that they do not span the
whole space of many-magnon ground states. Additional states were briefly discussed in other works \cite{PTPS.160.361,PhysRevB.70.100403} and numerical evidence for their contribution to
the ground state degeneracy was given \cite{0305-4470-39-34-006,10.1063/1.2780166}.

We should note, however, that there is no one to one correspondence of our model and the spin models for all considered lattices. A mapping of the spin model onto the Bose-Hubbard model
using a Holstein-Primakoff transformation \cite{PhysRev.58.1098} generates nearest-neighbor interactions which we do not take into account in our model. Nevertheless, in some cases
like e.g. the kagome lattice, the ground states of our model are equivalent to the ones for the spin model.

These spin models are as well a possible experimental realization of our model. A material exhibiting spins on a kagome geometry was presented in Ref. \cite{PhysRevB.49.3975}.
Another direct implementation of the bosonic Hubbard model lies in the field of ultracold atoms in optical lattices \cite{PhysRevLett.81.3108,nature415039a}. Damski \emph{et al.} \cite{PhysRevA.72.053612}
show how an optical kagome lattice can be generated, Eckardt \emph{et al.} \cite{0295-5075-89-1-10010} give ideas how to reverse the sign of the hopping matrix elements.

The purpose of this work is to generalize existing findings for few special lattices to all two-dimensional lattices that are line graphs and to characterize the full
ground state manifold of the model. We tackle these questions by using some notions of graph theory which have proven useful in the investigation of flat-band ferromagnetism in
the Hubbard model on line graphs \cite{0305-4470-24-14-018,0305-4470-24-2-005,0305-4470-25-16-011}.

The present paper is organized as follows. In Sec. \ref{sec:model}, we recall some definitions from graph theory and introduce our Hamiltonian and the considered lattice geometries.
The main theorem characterizing the ground states is stated in Sec. \ref{sec:theorem}. In Sec. \ref{sec:proof}, we give the proof of the main theorem. We first prove the linear independence
of the states defined in Sec. \ref{sec:theorem} and then show that they span the whole ground state manifold.

\section{The Model}
\label{sec:model}

In this section, we give some definitions concerning graphs and their line graphs. Furthermore, we introduce the Hamiltonian of our model and the basis of the space of
single-particle ground states.

\subsection{Definitions from graph theory}

In order to define our model, we need several notions of graph theory \cite{bollobas1998}, which we briefly review here. Let \(G\) be a graph. \(V(G)\) is the set
of vertices and \(E(G)\) the set of edges. Each edge \(e \in E(G)\) is an unordered pair of vertices and may be denoted as \(e = \{x,y\}\), \(x,y \in V(G)\). The edge \(e\)
is said to join the vertices \(x\) and \(y\). A walk of length \(n-1\) is an alternating sequence \(w = (x_1,e_1, \ldots , x_{n-1}, e_{n-1},x_n)\) of vertices and edges where
\( e_i = \{ x_i,x_{i+1} \} \). A path is a self-avoiding walk, i.e. \(x_i \neq x_j \) for \(i \neq j \). A cycle is a closed, self-avoiding walk, 
i.e. \(x_i \neq x_j \) for \(i \neq j\), \(j < n \) and \(x_1 = x_n\).

A graph \(G\) is connected, if for every pair of vertices \(\{x,y\}\) of \(G\), there exists a path from \(x\) to \(y\). A graph is said to be $k$-connected, if it contains no set of $k-1$ edges so
that \(G\) decays into two unconnected subgraphs, if these edges are deleted. As already stated in the abstract, we require the class of graphs we consider to be 2-connected, which clearly implies $k$-connectedness with $k>2$. It will become clear later that we need 2-connectedness since then each edge \(e \in E(G)\) is part of a cycle.

A graph is bipartite, if its vertex set \(V\) is the union of two disjoint sets \(V_1\) and \(V_2\) so that
each edge joins a vertex of \(V_1\) to a vertex of \(V_2\). The length of each cycle in a bipartite graph is even. A planar graph is a graph that can be drawn in the plane
so that no edges intersect each other. This representation of a planar graph is called a plane graph. It is not necessarily unique. However, the results we obtain do not
depend on the representation. In the following we consider bipartite 2-connected plane graphs.

A plane graph divides the plane into a set of connected components, called faces. Each plane graph has exactly one unbounded face and several bounded faces. The boundary of a face
is a cycle. We write \( f \cap g  = \emptyset \), if the cycles around \(f\) and \(g\) have no edges in common. Let \(F(G)\) be the set
of bounded faces of \(G\). Euler's Theorem relates the number of faces, vertices and edges of a plane graph by \( |F(G)| = |E(G)| - |V(G)| +1 \). 
We now define the line graph \(L(G)\) of a graph \(G\).

\begin{defin}
The line graph \(L(G)\) of a graph \(G\) is a graph whose vertex set \(V_L(G)\) is the edge set \(E(G)\) of the original graph \(G\). Two vertices 
\(e, e' \in E(G) \) of \(L(G)\) are joined by an edge, if \( |e \cap e'| = 1 \), i.e. if the edges in the original graph \(G\) have a vertex in common.
\end{defin}
The line graph of a plane graph is not necessarily a plane graph itself. Two examples of common line graphs are given in Fig.~\ref{fig:linegraph}.

\begin{figure}[h]
\centering
\subfloat[Honeycomb lattice (plane graph) with kagome lattice (line graph).]{\label{fig:kagome}\includegraphics[width=7.5cm]{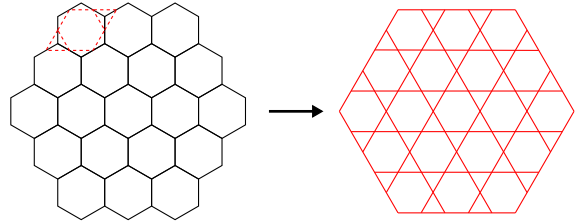}}
\\
\subfloat[Square lattice (plane graph) with checkerboard lattice (line graph).]{\label{fig:square}\includegraphics[width=7.5cm]{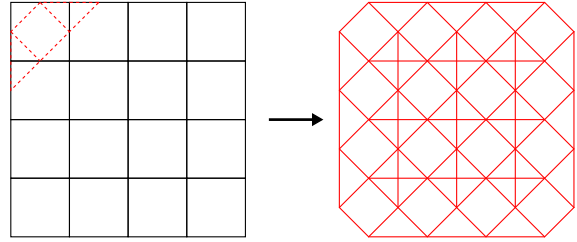}}
\caption[Lattices with their line graphs.]{\label{fig:linegraph} Examples of common plane graphs with their respective line graphs. The construction of the line graph is indicated in the upper left corners of the plane graphs.}
\end{figure}

Let \(B = (b_{xe})_{x \in V(G), e \in E(G)} \) be the incidence matrix of \(G\) with \(b_{xe} = 1 \), if \( x \in e \) and \(b_{xe} = 0 \) otherwise. The adjacency matrix
\(A = (a_{xy})_{x, y \in V(G)} \) is defined as follows: \(a_{xy} = 1 \), if \(x\) and \(y\) are adjacent, i.e. if the edge \(\{x,y\}\) exists and \(a_{xy} = 0 \) otherwise. The
adjacency matrix \(A_L\) of \(L(G)\) is related to the incidence matrix \(B\) of \(G\) by
\begin{equation}
\label{eq:adjacency}
 B^T B = 2 \, \mathbb{1}_{|E(G)|} + A_L,
\end{equation}
where \(\mathbb{1}_{|E(G)|}\) is the identity matrix of dimension \(|E(G)|\).

\subsection{The Hamiltonian}

We consider a bosonic Hubbard model on \(L(G)\) with the Hamiltonian
\begin{equation}
\label{eq:model}
 H = \sum_{e,e' \in E(G)} t_{ee'} b^\dagger_e b_{e'} + \sum_{e \in E(G)} U_e b^\dagger_e b^\dagger_e b_e b_e,
\end{equation}
where
\begin{equation}
 \label{eq:hopping}
 t_{ee'} = t \sum_{x \in V(G)} b_{xe} b_{xe'}.
\end{equation}
with \(b_{xe}/b_{xe'}\) being the elements of the incidence matrix \(B(G)\).
The operators \(b^\dagger_e\) and \(b_e\) are the usual creation and annihilation operators for bosons obeying the canonical commutation relations
\begin{equation}
 [b^\dagger_e,b^\dagger_{e'}] = [b_e,b_{e'}] = 0 \qquad \text{and} \qquad [b_e, b^\dagger_{e'}] = \delta_{ee'}.
\end{equation}
We choose \(U_e > 0 \) and \(t > 0\). Then the first part of the Hamiltonian describes the hopping of particles on \(L(G)\) whereas the second part is a repulsive on-site interaction.
\(t_{ee'}\) is the usual nearest-neighbor hopping plus an additional energy \(2t \delta_{ee'}\). 


\subsection{Single-particle basis}

Since \( B^T B \) is positive-semidefinite, the eigenvalues of \( \sum_{x \in V(G)} b_{xe} b_{xe'} \) are non-negative. It was shown by Mielke \cite{0305-4470-24-2-005}
that for bipartite 2-connected graphs,  \( B^T B \) has zero as an eigenvalue with multiplicity \( |E(G)|-|V(G)|+1 \). This is necessarily
the lowest eigenvalue. Each edge in \(E(G)\) may be oriented. With \(G\) being bipartite we may choose the orientation of each edge such that it points from \(V_1 (G)\)
to \(V_2 (G)\).
Furthermore, the boundary of each face in \(F(G)\) may be oriented. We choose a clockwise orientation. Now we may introduce the matrix \(S = (s)_{f \in F(G), e \in E(G)}\) as
\begin{equation}
 s_{fe} = \begin{cases}
           1,  & \text{if }e\text{ belongs to the boundary of }f \\ & \text{and points into the direction of }f, \\
           -1, & \text{if }e\text{ belongs to the boundary of }f \\ & \text{and points into the opposite direction of }f, \\
           0,  & \text{otherwise.}
          \end{cases}
\end{equation}
We have \(B S^T = 0\) and \(\dim\left(\ker(B)\right) = |F(G)| \). Hence, the columns of \(S^T\) form a basis of the kernel of \(B\). For each face \(f\) we define operators
\begin{equation}
\label{eq:face_op}
 b^\dagger_f = \sum_{e \in E(G)} s_{fe} b^\dagger_e.
\end{equation}
Then the states \( b^\dagger_f \left| 0 \right\rangle \) span the space of single-particle ground states with kinetic
energy zero. The particle is delocalized on the edges of the boundary cycle of the face \(f\) which are vertices of \(L(G)\). 
Let us examine the properties of cycles on \(G\) in this context. Since a cycle is self-avoiding, it encloses a simply connected part of the plane. This part is a union of faces.
If we denote the set of faces enclosed by a cycle \(c\) as \(F(c)\), we may see the following: The state \( \sum_{f \in F(c)} b^\dagger_f \left| 0 \right\rangle \) is a single-particle
ground state of \(H\) since it is a linear combination of basis states. The particle is delocalized only on the vertices of \(L(G)\) that are edges of \(c\). All other creation operators
contained in the different \( b^\dagger_f \) are canceled due to the alternating signs of the \(s_{fe}\) for adjacent faces.
In the following, we always deal with cycles on \(G\). The corresponding states and the model, however, are defined on the line graph \(L(G)\). This is illustrated in Fig.~\ref{fig:state}.

\begin{figure}[h]
 \centering
 \includegraphics[width=8.5cm]{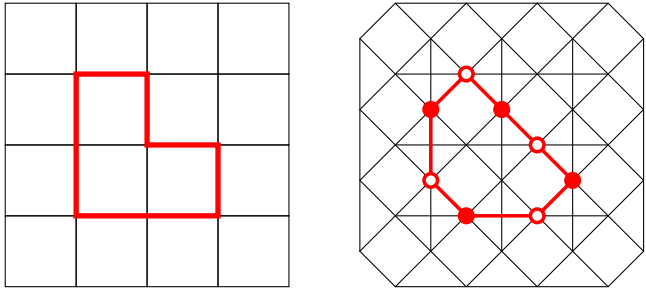}
 \caption{A cycle on a plane graph (left) and its corresponding single-particle ground state of \(H\) on the line graph (right). The alternating amplitudes caused by the factor \(s_{fe}\)
 in Eq.~\eqref{eq:face_op} are indicated by filled and empty circles.}
 \label{fig:state}
\end{figure}

\section{Main theorem}
\label{sec:theorem}

In this section, we state our main theorem. The local single-particle basis allows us to construct multi-particle states that have both kinetic and interaction
energy zero. We define a set of states that we identify as a basis of the Fock space \(\mathcal{F}_0\) of multi-particle ground states.

\subsection{Cycle sets}

Before we are able to state the theorem, we have to introduce the notion of cycle sets. A cycle set \(C = \{c_i, \ i = 1, \ldots ,N\}\) is a set of edge-disjoint cycles 
\(c_i = (x_{1i}, e_{1i}, \ldots , e_{n-1i} , x_{ni}), \ x_{ni}=x_{1i} \). The cycles in a cycle set may have vertices in common.
We now define the contraction of a cycle set. Any cycle \(c'\) with
\(F(c') \subset F(c)\) is called a contraction of \(c\). A cycle set \(C' = \{c'_i, \ i = 1, \ldots ,N\} \) is called a contraction of \(C\), if \( F(c'_i) \subseteq F(c_i) \) for all
\(i = 1, \ldots ,N\) and \(\bigcup_i F(c'_i) \subset \bigcup_i F(c_i) \). A cycle set is called contractible, if a contraction exists and uncontractible otherwise. Examples of contractions
of cycle sets are given in Fig.~\ref{fig:cont}.

\begin{figure}[h]
\centering
\subfloat[]{\label{fig:cont_hex}\includegraphics[height=5cm]{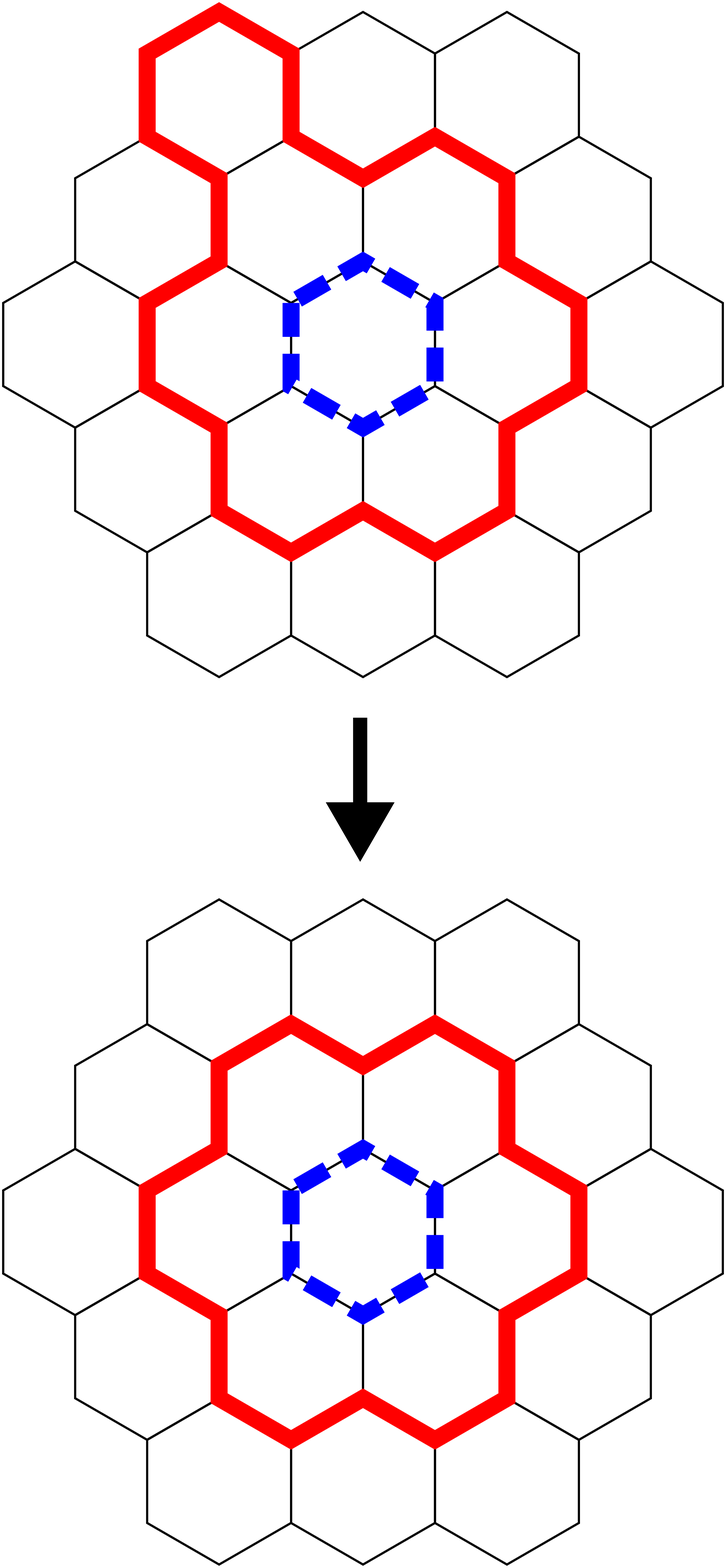}}
\hspace{2cm}
\subfloat[]{\label{fig:cont_square}\includegraphics[height=5cm]{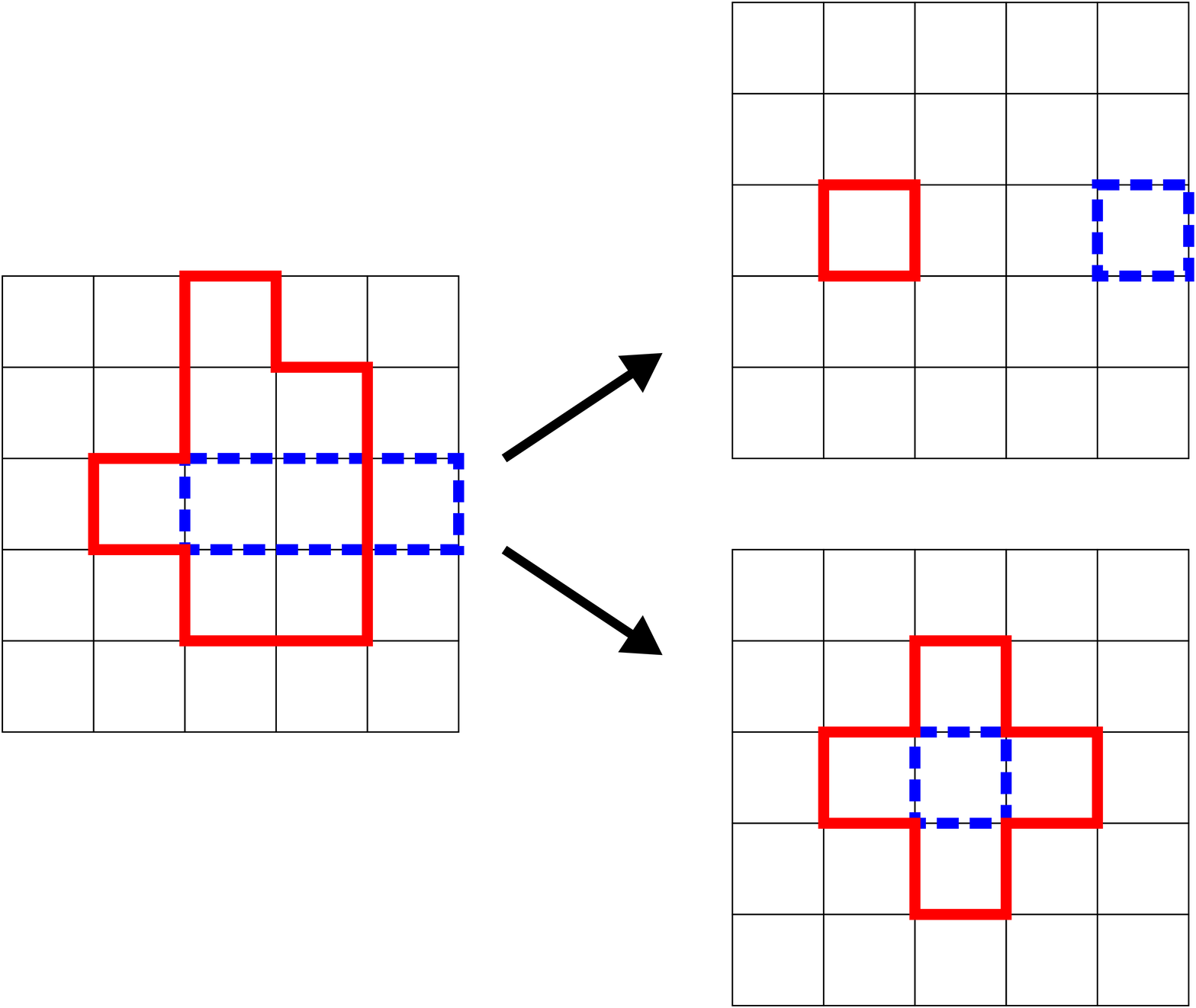}}
\caption[Contractions of cycle sets.]{Examples of cycle sets with two elements and their uncontractible contractions. In case \subref{fig:cont_hex}, the uncontractible contraction is unique. In case
\subref{fig:cont_square}, different uncontractible contractions may be obtained from the cycle set. Two of these are depicted.}
\label{fig:cont}
\end{figure}

Furthermore, we define the notion of non-intersecting cycles. Two edge-disjoint cycles \(c_1\) and \(c_2\) are non-intersecting, if either one is the contraction of the other or if
\( f \cap g = 0 \) for all \( f \in F(c_1) \) and \( g \in F(c_2) \). We remark that one of the two conditions is sufficient, whereas the cycles have to be edge-disjoint.
If two cycles \(c_1\) and \(c_2\) are non-intersecting,
we write \( c_1 \circledcirc c_2\). The two cycles in Fig.~\ref{fig:cont}\subref{fig:cont_hex} are non-intersecting. The cycles in the left part of
Fig.~\ref{fig:cont}\subref{fig:cont_square} are intersecting. The right part of Fig.~\ref{fig:cont}\subref{fig:cont_square} is depicting non-intersecting cycles which have vertices in
common.

\subsection{Theorem}

We now associate multi-particle states to each cycle set \(C\).
\begin{equation}
 \left| \Phi (C) \right\rangle = O^\dagger (C) \left| 0 \right\rangle,
\end{equation}
with
\begin{equation}
\label{eq:operator}
 O^\dagger (C) = \prod_{c_i \in C} \sum_{f \in F(c_i)} b^\dagger_f.
\end{equation}
Since the cycles in cycle sets are edge-disjoint, \( \left| \Phi (C) \right\rangle \) contains no doubly occupied sites on \( L(G) \). The states are ground states minimizing
both the kinetic and interaction energy simultaneously.
Now we are able to state our main theorem.
\begin{theorem}
 The multi-particle states \( \left| \Phi (C^{(u)}) \right\rangle \) belonging to uncontractible cycle sets \(C^{(u)}\) form a basis of the Fock space \(\mathcal{F}_0\) of all the
kernel of \(H\).
\end{theorem}

\section{Proof}

\label{sec:proof}

We begin this section by defining the states \( \left| \Phi (C^{(u)}) \right\rangle \) in a different way which is more convenient throughout the proof. This definition allows us
to construct the states explicitly. The states are recursively defined.

\subsection{Alternative definition of states}

Let us denote the set of all \( \left| \Phi (C^{(u)}) \right\rangle \) as \(\mathcal{B}\). \(\mathcal{B}\) is the
union of all sets \(\mathcal{B}_N\) of \( \left| \Phi (C^{(u)}_N) \right\rangle \) with fixed particle number \(N\).
\begin{equation}
 \mathcal{B} = \bigcup_{N=0}^\infty \mathcal{B}_N.
\end{equation}
Since the uncontractible cycle sets with \( N=1\) are the sets containing a boundary cycle of a face \(f \in F(G) \), the set \(\mathcal{B}_1\) is the set of single-particle basis states.
\begin{equation}
 \mathcal{B}_1 = \left\{ b^\dagger_f \left| 0 \right\rangle, \ f \in |F(G)| \right\}.
\end{equation}
The set of operators creating these states is called \(\mathcal{O}_1\).
\begin{equation}
 \mathcal{O}_1 = \left\{ b^\dagger_f , \ f \in |F(G)| \right\}.
\end{equation}
The set of operators creating the states of \(\mathcal{B}_N\) is called \(\mathcal{O}_N\). As we know from \eqref{eq:operator}, an element of \(\mathcal{O}_N\) has the following form:
\begin{equation}
\label{eq:operator_uc}
 O^\dagger(C^{(u)}_N) = \prod_{i=1}^N \sum_{f \in F(c_i)} b^\dagger_f.
\end{equation}
We now introduce disjoint subsets 
\begin{equation}
 C^{(u)}_N (k) = \left\{ c_i, \ i \in I_k \subset \{1, \ldots, N\} \right\}
\end{equation}
of the cycle set \( C^{(u)}_N \) which are defined as follows:
\begin{itemize}
 \item For each pair \(c_i \in C^{(u)}_N (k) \), \(c_j \in C^{(u)}_N (l) \) with \( k \neq l \), we may not find three faces \(f \in F(c_i)\), \( g \in F(c_j) \) and \(h \in F(G)\)
       with \(f \cap h \neq \emptyset \) and \(g \cap h \neq \emptyset \). 
 \item In case of \( | C^{(u)}_N (k) | \geq 2 \), there has to exist a \(c_j \in C^{(u)}_N (k)\) for each \(c_i \in C^{(u)}_N (k) , \ i \neq j \), with the following property: We
       may find three (not necessarily different) faces \(f \in F(c_i)\), \( g \in F(c_j) \) and \(h \in F(G)\) with \(f \cap h \neq \emptyset \) and \(g \cap h \neq \emptyset \).
\end{itemize}
In simple words, the cycles in one subset have at most one face between each other. The cycles in different subsets have at least two faces between each other. 
An example of such a partition is illustrated in Fig.~\ref{fig:subset}. We call the subsets close-packed subsets.

\begin{figure}[h]
\centering
\subfloat[]{\label{fig:set1}\includegraphics[width=5cm]{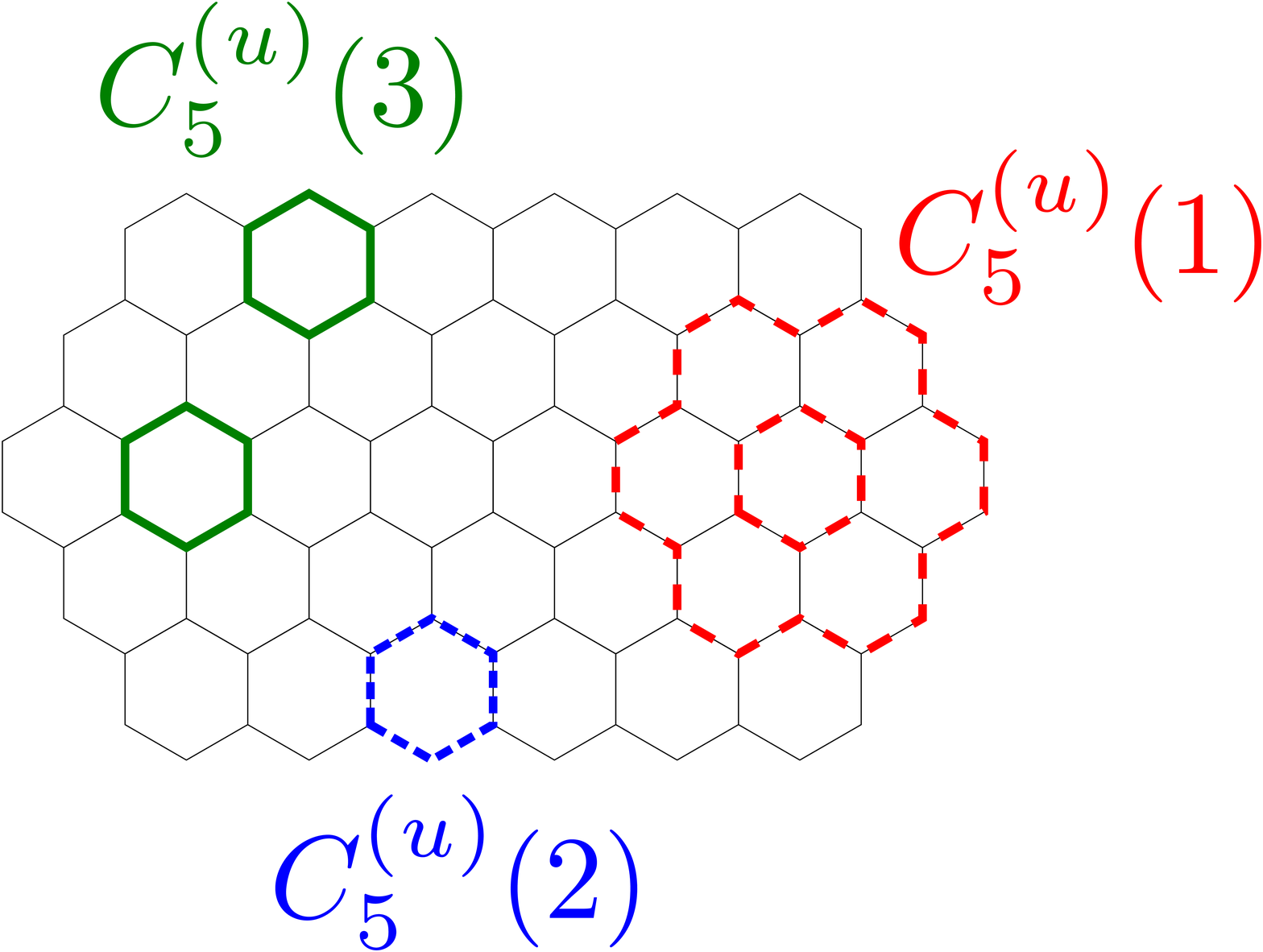}}
\hspace{1cm}
\subfloat[]{\label{fig:set2}\includegraphics[width=5cm]{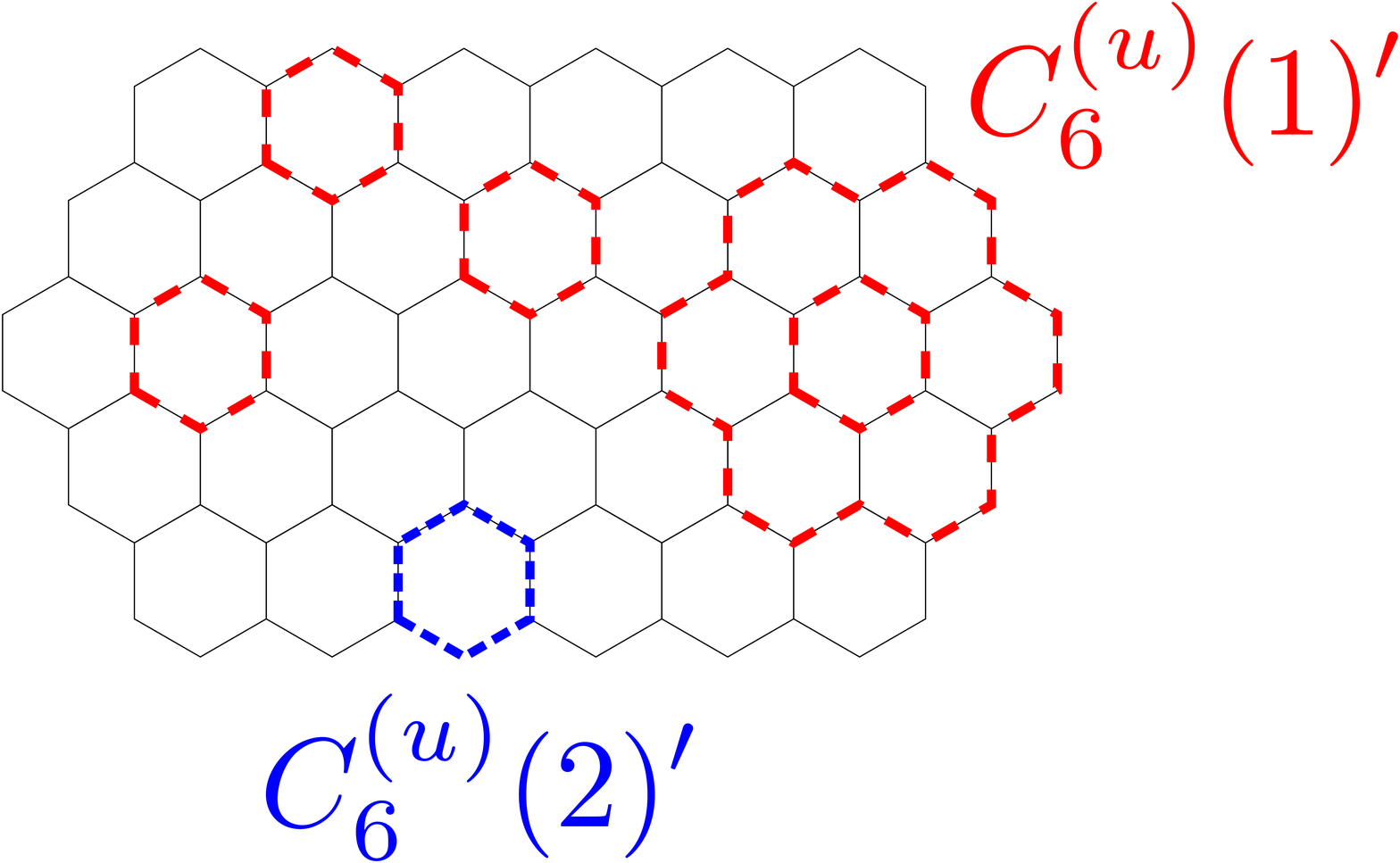}}
\caption[Close-packed subsets of cycle sets.]{Examples of uncontractible cycle sets and their partition into close-packed subsets. The cycle set in \subref{fig:set1} consists of five cycles and
is divided into three close-packed subsets. The one in \subref{fig:set2} contains six cycles but only two close-packed subsets. The additional cycle in part \subref{fig:set2} ``unifies'' the
subsets \( C_5^{(u)}(1) \) and \( C_5^{(u)}(3) \) to the subset \( C_6^{(u)}(1)' \).}
\label{fig:subset}
\end{figure}

We now introduce a set \(\mathcal{\tilde B}_N\) which is constructed from the states \( O^\dagger(C^{(u)}_{N-1}) \left| 0 \right\rangle \in \mathcal{B}_{N-1}\). It contains all
states of the following two forms.
\begin{enumerate}
 \item States that are created by adding a new particle on a boundary cycle of a face \(f\) which has no edge in common with all faces enclosed by cycles of \( C^{(u)}_{N-1} \):
       \begin{equation}
        \left| \Phi \right\rangle = O^\dagger(C^{(u)}_{N-1}) \, b_f^\dagger \left| 0 \right\rangle,
       \end{equation}
       with \( f \cap g = \emptyset \) for all \( g \in \bigcup_{c \in C^{(u)}_{N-1}} F(c) \).
 \item States that are created by adding a new particle on a cycle around a set of close-packed cycles:
       \begin{equation}
        \left| \Phi \right\rangle = O^\dagger(C^{(u)}_{N-1}) \sum_{f \in F(G)} a_{fk} \, b_f^\dagger \left| 0 \right\rangle,
       \end{equation}
       where
       \begin{equation}
        a_{fk} = \begin{cases}
                 1, & \text{if }f \cap g \neq \emptyset \text{ for any } g \in \bigcup_{c \in C^{(u)}_{N-1} (k)} F(c), \\
                 0, & \text{otherwise.}
                 \end{cases}
       \end{equation}
       Furthermore, \( g \cap f_0 = \emptyset \) has to be fulfilled for all \(g \in \bigcup_{c \in C^{(u)}_{N-1} (k)} F(c) \), where the unbounded face is denoted by \(f_0\).
\end{enumerate}
We remark that all states in \(\mathcal{\tilde B}_N\) may be created from up to \(N\) elements of \(\mathcal{B}_{N-1}\). Of course, these states are incorporated into
\(\mathcal{\tilde B}_N\) only once. The following proposition relates the abstract notion of states associated to uncontractible cycle sets to the states explicitly constructed by
the above procedure.
\begin{prop}
\label{prop:b_tilde}
 With the above construction we obtain \( \mathcal{\tilde B}_N = \mathcal{B}_N \).
\end{prop}
\begin{proof}[Proof of proposition \ref{prop:b_tilde}]
Any state of \( \mathcal{\tilde B}_N \) is a state that may be written as \( O^\dagger (C_N) \left| 0 \right\rangle \). It is clear that the cycle on which the \(N\)-th particle
is added cannot be contracted further. Since the other cycles belong to an uncontractible cycle set, the cycle set \( C_N \) is also uncontractible. \( \mathcal{B}_{N-1} \) contains
states associated to all possible uncontractible cycle sets with \(N-1\) cycles. In the construction of \( \mathcal{\tilde B}_N \), all possible \(N\)-th uncontractible cycles are
added to a cycle set of \( \mathcal{B}_{N-1} \). Hence, \( \mathcal{\tilde B}_N \) contains all states associated to uncontractible cycle sets consisting of \(N\) cycles and 
\( \mathcal{\tilde B}_N = \mathcal{B}_N \).
\end{proof}
The states in \( \mathcal{B}_N \) constructed recursively from \( \mathcal{B}_1 \) by the above procedure are all states associated to cycle sets with pairwise non-intersecting cycles.
This leads to 
\begin{cor}
 All uncontractible cycle sets consist of pairwise non-intersecting cycles.
\end{cor}

We remark the following concerning the construction of \( \mathcal{B}_N \). If no particle can be added to any state in \(\mathcal{B}_{N-1}\) on an additional cycle obeying the rules
of construction given above, \( \mathcal{B}_N \) will be empty. Then our construction is no longer suitable and we have reached the critical filling factor.

\subsection{Linear independence}

We now prove the linear independence of the states in \( \mathcal{B}_N \). We show that each state \( \left| \Phi (C_N^{(u)}) \right\rangle \in \mathcal{B}_N \)
contains a combination of creation operators that may not be found in any other state of \( \mathcal{B}_N \). This is a sufficient property for the states to be linearly independent.
The \(k\)-th summand of \( \left| \Phi (C^{(u)}_N) \right\rangle \) may be written as
\begin{equation}
 \left| \varphi (k) \right\rangle = \prod_{j=1}^N \left(b^\dagger_{f_j} \right)^{n_j^k} \left| 0 \right\rangle,
\end{equation}
with \(f_j \in F(G)\), \( n_j^k \in \{0, \dots, N\}\), \( n_1^k \geq \dots \geq n_N^k\) and \(\sum_j n_j^k = N \). We drop the label \( C^{(u)}_N \) for the sake of
readability. The operators are ordered in such a way that their exponents are decreasing with increasing \(j\). If we recall the definition of \( \left| \Phi (C^{(u)}_N) \right\rangle
= O^\dagger(C^{(u)}_N) \left| 0 \right\rangle \), we can see that expanding the product in \eqref{eq:operator_uc} leads to \(M = \prod_{i=1}^N |F(c_i)| \) summands.  We denote the
index set \( \{ 1, \ldots , M \} \) of the individual summands as \( K_0 \). Let us now introduce subsets \(K_j \subset K_0\) defined as follows:
\begin{equation}
 K_j = \left\{ k \ | \  k = \operatorname*{\arg\max}_{k \in K_{j-1}} \left( n^k_j\right) \right\}.
\end{equation}
We choose an arbitrary summand \(\left| \varphi (k^\ast) \right\rangle\) with \(k^\ast \in K_N\). This leads to the following proposition.
\begin{prop}
\label{prop:independence}
 The summand \(\left| \varphi (k^\ast) \right\rangle\) is exclusively contained in \( \left| \Phi (C^{(u)}_N) \right\rangle \), i.e. it uniquely determines
 \( \left| \Phi (C^{(u)}_N) \right\rangle \).
\end{prop}

\begin{proof}[Proof of proposition \ref{prop:independence}]
Let us construct a state \( \left| \Phi (C^{(u)}_N) \right\rangle \) containing a given \(\left| \varphi (k^\ast) \right\rangle\). In order to have \(f_1\) \(n_1^k\)-fold occupied,
we put one particle on \(f_1\) and \(n_1^k-1\) on cycles around it. We proceed in the same manner with \(f_2\). The already existing cycles around \(f_1\) may have to be enhanced
to avoid two cycles having edges in common. We continue the procedure over all occupied faces in \(\left| \varphi (k^\ast) \right\rangle\). If any of the steps is not possible
according to the given rules, \(\left| \varphi (k^\ast) \right\rangle\) is not a state fulfilling the assumptions of proposition \ref{prop:independence}. This is the only procedure
to construct an \(N\)-particle state of the set \(\mathcal{B}_N\) containing \(\left| \varphi (k^\ast) \right\rangle\). It leads to a unique state
\( \left| \Phi (C^{(u)}_N) \right\rangle \).
\end{proof} 

\subsection{Spanning property}

We now proceed by showing that each ground state with \(N\) particles may be written as a linear combination of states of \(\mathcal{B}_N\). The proof is divided into three steps.
First, we show that every \(N\)-particle ground state may be expressed as a linear combination of states \( \left| \Phi (C) \right\rangle \) that are associated to a not
necessarily uncontractible cycle set. In the second step, we demonstrate that each of the states \( \left| \Phi (C) \right\rangle \) may be decomposed into a sum of states
\( \left| \Phi (D^{(n)}) \right\rangle \) where the \(D^{(n)}\) are cycle sets with pairwise non-intersecting cycles. Finally, we prove that a state associated to a cycle set
with this property may be written as a sum of elements of \(\mathcal{B}_N\).

Each \(N\)-particle state with kinetic energy zero may be written as
\begin{equation}
\label{eq:general}
 \left| \Phi \right\rangle = \sum_{f_1,\ldots,f_N}{\Phi \left(f_1,\ldots,f_N\right)b^\dagger_{f_1}\ldots b^\dagger_{f_N}} \left| 0 \right\rangle.
\end{equation} Any state \(\left| \Phi \right\rangle \) is a ground state of \(H\) if and only if there is no doubly occupied site:
\begin{equation}
\label{eq:condition}
 b_e^2 \left| \Phi \right\rangle = 0 \qquad \text{for all }e \in E(G).
\end{equation}

\begin{prop}
\label{prop:noedge}
 Each ground state \( \left| \Phi_G \right\rangle \) of \(H\) on \(L(G)\) with \(N\) particles may be written as a linear combination of states \( \left| \Phi (C_i) \right\rangle \),
 where the \(C_i\) are cycle sets.
 \begin{equation}
 \label{eq:prop3}
  \left| \Phi_G \right\rangle = \sum_i \Phi_i \left| \Phi (C_i) \right\rangle = \sum_i \Phi_i \prod_{j=1}^N \sum_{f \in F(c_{ij})} b_f^\dagger \left| 0 \right\rangle.
 \end{equation}
 Here, \(c_{ij}\) is the \(j\)-th cycle of the \(i\)-th cycle set.
\end{prop}

\begin{proof}[Proof of proposition \ref{prop:noedge}]
Let us examine the constraint \eqref{eq:condition} for an edge \(e\) belonging to the faces \(f\) and \(g\). If we recall the definition of the \(b_f^\dagger\) \eqref{eq:face_op},
we find the commutator \([b_e, b_f^\dagger] = s_{fe}\). Without loss of generality, we may set \([b_e, b_f^\dagger] = 1\) and \([b_e, b_g^\dagger] = -1\). We now consider the sum of all
terms of \(\left| \Phi \right\rangle\) in \eqref{eq:general} containing \(n\) operators \(b_f^\dagger\) or \(b_g^\dagger\) and an arbitrary but fixed combination
\(b^\dagger_{f_1} \ldots b^\dagger_{f_m}\).
\begin{equation}
 \left| \Phi_n^e \right\rangle = \sum_{l=0}^n \Phi_{nl} \, b^\dagger_{f_1} \ldots b^\dagger_{f_m} \left( b^\dagger_f \right)^l \left( b^\dagger_g \right)^{n-l} \left| 0 \right\rangle,
\end{equation}
with \(n + m = N\) and \( b^\dagger_{f_1} \ldots b^\dagger_{f_m} \neq f,g \). Applying condition \eqref{eq:condition} we obtain
 \begin{eqnarray}
  \fl b_e^2 \left| \Phi_n^e \right\rangle = \ b^\dagger_{f_1} \ldots b^\dagger_{f_m} \sum_{l=0}^n \Phi_{nl} & \left[ l(l-1) \left( b_f^\dagger \right)^{l-2} 
  \left( b_g^\dagger \right)^{n-l} - 2l(n-l) \left( b_f^\dagger \right)^{l-1} \left( b_g^\dagger \right)^{n-l-1} \right. \nonumber \\
  & \left.  + (n-l)(n-l-1) \left( b_f^\dagger \right)^{l} 
     \left( b_g^\dagger \right)^{n-l-2} \right] \left| 0 \right\rangle.
\end{eqnarray}
Since this is the only term containing \(n-2\) operators \( b_f^\dagger \) or \( b_g^\dagger \) plus the specific combination \( b^\dagger_{f_1} \ldots b^\dagger_{f_m} \), it is linearly
independent of all other terms in \( b_e^2 \left| \Phi \right\rangle \) and has to vanish separately. By substituting the summation indices, we obtain terms containing the same number of
\( b_f^\dagger \) and \( b_g^\dagger \) which are as well linearly independent of each other. This finally leads to the constraint
\begin{equation}
\label{eq:nl}
 l(l-1) \Phi_{nl} - 2(l-1)(n-l+1) \Phi_{nl-1} + (n-l+2)(n-l+1) \Phi_{nl-2} = 0.
\end{equation}
This means that only two of the \( \Phi_{nl} \) are arbitrary for a fixed \( n \). We show that we may write
\begin{equation}
\label{eq:fg}
 \left| \Phi_n^e \right\rangle = b^\dagger_{f_1} \ldots b^\dagger_{f_m} \left( \Phi_n^f \, b^\dagger_f \left( b^\dagger_f + b^\dagger_g \right)^{n-1} 
  + \Phi_n^g \,  b^\dagger_g \left( b^\dagger_f + b^\dagger_g \right)^{n-1} \right) \left| 0 \right\rangle.
\end{equation}
The expression is valid if

\begin{equation}
 \Phi_{nl} = \binom {n-1} {l-1} \Phi_n^f + \binom {n-1} {l} \Phi_n^g.
\end{equation}
Let us prove this by induction on \(l\). We set \( \Phi_{n0} = \Phi_n^g \) and \( \Phi_{n1} = (n-1) \Phi_n^g + \Phi_n^f \) and assume that the statement holds for \(l-2\) and \(l-1\).
If we solve \eqref{eq:nl} for \(\Phi_{nl}\), it reads
\begin{eqnarray}
\fl \Phi_{nl} & = &\left[ \frac{2}{l}(n-l+1) \binom {n-1} {l-2} - \frac{(n-l+2)(n-l+1)}{l(l-1)} \binom {n-1} {l-3} \right] \Phi_n^f \nonumber \\
              &   &+ \left[ \frac{2}{l}(n-l+1) \binom {n-1} {l-1} - \frac{(n-l+2)(n-l+1)}{l(l-1)} \binom {n-1} {l-2} \right] \Phi_n^g \nonumber \\
              & = &\binom {n-1} {l-1} \Phi_n^f + \binom {n-1} {l} \Phi_n^g,
\end{eqnarray}
which proves \eqref{eq:fg}.

If we reformulate \eqref{eq:general}, each state with kinetic energy zero may be written as
\begin{equation}
\label{eq:nocommon}
 \left| \Phi \right\rangle = \sum_{i \in I} \Phi_i \prod_{j=1}^N  \sum_{f \in F(c(i,j)) } b^\dagger_f \left| 0 \right\rangle,
\end{equation}
where \(I\) is an index set with \( |I| \leq \binom {|F(G)|+N-1} {N} \) and the \(c(i,j)\) are cycles on \(G\). Without loss of generality, we assume that
the summands \( \prod_{j=1}^N  \sum_{f \in F(c(i,j)) } b^\dagger_f \left| 0 \right\rangle \) are linearly independent. Note that the \(c(i,j)\) for fixed \(i\) do not necessarily form
cycle sets. The expression in \eqref{eq:nocommon} is not equal to \eqref{eq:prop3} in the proposition. We now claim that if \( \left| \Phi \right\rangle \) is a ground state, we have to be able
to order the right hand side of \eqref{eq:nocommon} in such a way that we have
\begin{equation}
\label{eq:knack}
 b_e \sum_{f \in F(c(i,j))} b^\dagger_f \left| 0 \right\rangle \neq 0 \ \Rightarrow \ b_e \sum_{f \in F(c(i,k))} b^\dagger_f \left| 0 \right\rangle = 0 
\end{equation}
for all \( k \neq j \) and for all \( e \in E(G)\). If \eqref{eq:knack} were not true, \( \prod_{j=1}^N  \sum_{f \in F(c(i,j)) } b^\dagger_f \left| 0 \right\rangle \)
would be at least doubly occupied on an edge \(e\). Let us assume that it is doubly occupied and \(e\) belongs to the cycles \(c(i,1)\) and \(c(i,2)\). 
The double occupancies are caused by \(b^\dagger_f\) on \(c(i,1)\)
and \(b^\dagger_g\) on \(c(i,2)\), respectively. The faces \(f\) and \(g\) are either neighbors or equal. There are only four ways to remove all terms containing
\( \left( b^\dagger_e \right)^2 \):
\begin{itemize}
 \item Remove \(f\) from \(F(c(i,1))\),
 \item add the face adjacent to \(f\) to \(F(c(i,1))\),
 \item remove \(g\) from \(F(c(i,2))\),
 \item add the face adjacent to \(g\) to \(F(c(i,2))\).
\end{itemize}
This can be seen from Eq. \eqref{eq:fg} and the fact that the individual summands of \eqref{eq:nocommon} are linearly independent.
Hence, in order to obtain a ground state, the double occupancy has to be removed by a combination of the four procedures mentioned above. The terms to do this have to be contained
in \( \left| \Phi \right\rangle \). This contracts or widens the cycles and \eqref{eq:knack} for \(e\) holds again for each summand. The argument is analogous for a higher than
double occupancy. Then several \(F(c)\) have to be changed.
Since the above statements hold for each edge of \(G\), \eqref{eq:knack} has to be fulfilled, if \( \left| \Phi \right\rangle \)
is a ground state. We remark that the representation of \eqref{eq:nocommon} is not necessarily unique.
If \eqref{eq:knack} holds, the \( c(i,j) \) are edge-disjoint for fixed \(i\). Thus, the set \(\{ c(i,j), \ j = 1, \ldots , N \}\) is indeed a cycle set, if 
\( \left| \Phi \right\rangle \) is a ground state. This finally proves proposition \ref{prop:noedge}.
\end{proof}

We have now established that each \(N\)-particle ground state may be written as a linear combination of states associated to cycle sets. We may now proceed with the second step of
our proof. Let us consider these states associated to a cycle set \(C\).
\begin{equation}
\label{eq:cycle_set}
 \left| \Phi(C) \right\rangle = \prod_{k=1}^N \sum_{f \in F(c_k)} b_f^\dagger \left| 0 \right\rangle.
\end{equation}
They may be further decomposed which leads to the following proposition.
\begin{prop}
 \label{prop:int}
Each state of the form \eqref{eq:cycle_set} may be written as a sum of states associated to cycle sets with pairwise non-intersecting cycles.
\begin{equation}
 \left| \Phi(C) \right\rangle = \sum_i \left| \Phi(D^{(n)}_i) \right\rangle = \sum_i \prod_{j=1}^N \sum_{f \in F(d_{ij})} b_f^\dagger \left| 0 \right\rangle.
\end{equation}
The \(D^{(n)}_i\) are cycle sets and \(d_{ij}\) is the \(j\)-th cycle of the \(i\)-th cycle set. The cycles in \(D^{(n)}_i\) are non-intersecting 
(\( d_{ij} \circledcirc d_{ik}\) for \( j \neq k \)). Furthermore, the states \( \left| \Phi(D^{(n)}_i) \right\rangle \) only have occupied edges
where \(\left| \Phi(C) \right\rangle\) has occupations.
\end{prop}

\begin{proof}[Proof of proposition \ref{prop:int}]
Let us first consider the case of a cycle set \(C_2\) containing only two cycles \(c_1\) and \(c_2\).
\begin{equation}
 \left| \Phi(C_2) \right\rangle = \sum_{f \in F(c_1)} b^\dagger_f \sum_{f \in F(c_2)} b^\dagger_f  \left| 0 \right\rangle.
\end{equation}
We now assume that \(c_1\) and \(c_2\) are intersecting. The cycles divide the set \(F(c_1) \cup F(c_2)\) into disjoint subsets. Some of these subsets are enclosed by one of the
cycles, others by both. The union of faces in these subsets represents a simply connected part of the plane. Therefore their boundary is a cycle on \(G\). We call these boundary cycles
\( \alpha_i \) (enclosed by one cycle) and \( \beta_k \) (enclosed by both cycles). We remark that \( \sum_{f \in F(\alpha_i)} b^\dagger_f \left| 0 \right\rangle \) and 
\( \sum_{f \in F(\beta_i)} b^\dagger_f \left| 0 \right\rangle \) only have edges occupied which have been occupied in \(\left| \Phi(C_2) \right\rangle\).
In addition, \(\alpha_i\) and \(\alpha_j\) do not have edges in common for \( i \neq j \). The same holds for \(\beta_k\) and \(\beta_l\) for \( k \neq l \). We define
\(F_1 = \bigcup_{F(\alpha_i) \subset F(c_1)} F(\alpha_i) \), \(F_2 = \bigcup_{F(\alpha_i) \subset F(c_2)} F(\alpha_i) \) and \(F_{12} = \bigcup_k F(\beta_k) \). Then
\(\left| \Phi(C_2) \right\rangle\) may be written as
\begin{equation}
\label{eq:two_cycles}
 \left| \Phi(C_2) \right\rangle = \sum_{f \in F_1 \cup F_{12}} b^\dagger_f  \sum_{f \in F_2 \cup F_{12}} b^\dagger_f \left| 0 \right\rangle,
\end{equation}
with \( F_1 \cap F_2 = F_1 \cap F_{12} = F_2 \cap F_{12} = \emptyset \).

By expanding and rearranging the summands, \eqref{eq:two_cycles} reads
\begin{equation}
 \left| \Phi(C_2) \right\rangle = \left( \sum_{f \in F_1 \cup F_2 \cup F_{12}} b^\dagger_f \sum_{f \in F_{12}} b^\dagger_f
 + \sum_{f \in F_1} b^\dagger_f \sum_{f \in F_2 } b^\dagger_f \right) \left| 0 \right\rangle.
\end{equation}
This may be decomposed to
\begin{eqnarray}
\label{eq:intersect}
 \fl \left| \Phi(C_2) \right\rangle = &\left[ \sum_k \left( \sum_{f \in F(\beta_k)} b^\dagger_f \right) \left(  \sum_{l \neq k} \sum_{f \in F(\beta_l)} b^\dagger_f 
 + \sum_{j \notin I_k} \sum_{f \in F(\alpha_j)}  b^\dagger_f   \right. \right. \nonumber \\
 &\left. \left. + \sum_{f \in F(\beta_k) \cup \bigcup_{i \in I_k} F(\alpha_i)} b^\dagger_f \right)
   + \left( \sum_{F(\alpha_i) \subset F(c_1)} \sum_{f \in F(\alpha_i)} b^\dagger_f \right)
  \right. \nonumber \\
 &\left. \left( \sum_{F(\alpha_i) \subset F(c_2)} \sum_{f \in F(\alpha_i)} b^\dagger_f \right) \right] \left| 0 \right\rangle,
\end{eqnarray}
where \(I_k = \{i \ | \ \alpha_i \text{ has edges in common with } \beta_k \}\). Let us examine the two individual parts of this expression. The first term
describes the different combinations of a particle on a cycle \(\beta_k\) plus one on either
\begin{itemize}
 \item a cycle \( \beta_l \), or
 \item a cycle \(\alpha_j\) that has no edge in common with \(\beta_k\), or
 \item a cycle consisting of the edges that belong only the \(\alpha_i\) that have edges in common with \(\beta_k\).
\end{itemize}
The second term describes the
combinations of two particles on two cycles \(\alpha_i\). All of these combinations consist of two non-intersecting cycles. This shows that proposition \ref{prop:int} holds
for \(N=2\).

Figs.~\ref{fig:non_int} and \ref{fig:proof} illustrate an example of two intersecting cycles and their decomposition into non-intersecting cycles.

\begin{figure}[h]
\centering
 \includegraphics[width=8.5cm]{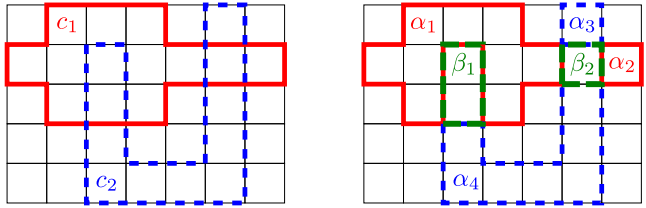}
\caption[Two non-intersecting cycles.]{Division of the plane by two intersecting cycles \(c_1\) and \(c_2\) and the definition of the cycles \(\alpha_i\) and \(\beta_k\).}
\label{fig:non_int}
\end{figure}

\begin{figure}[h]
\centering
 \includegraphics[width=8.5cm]{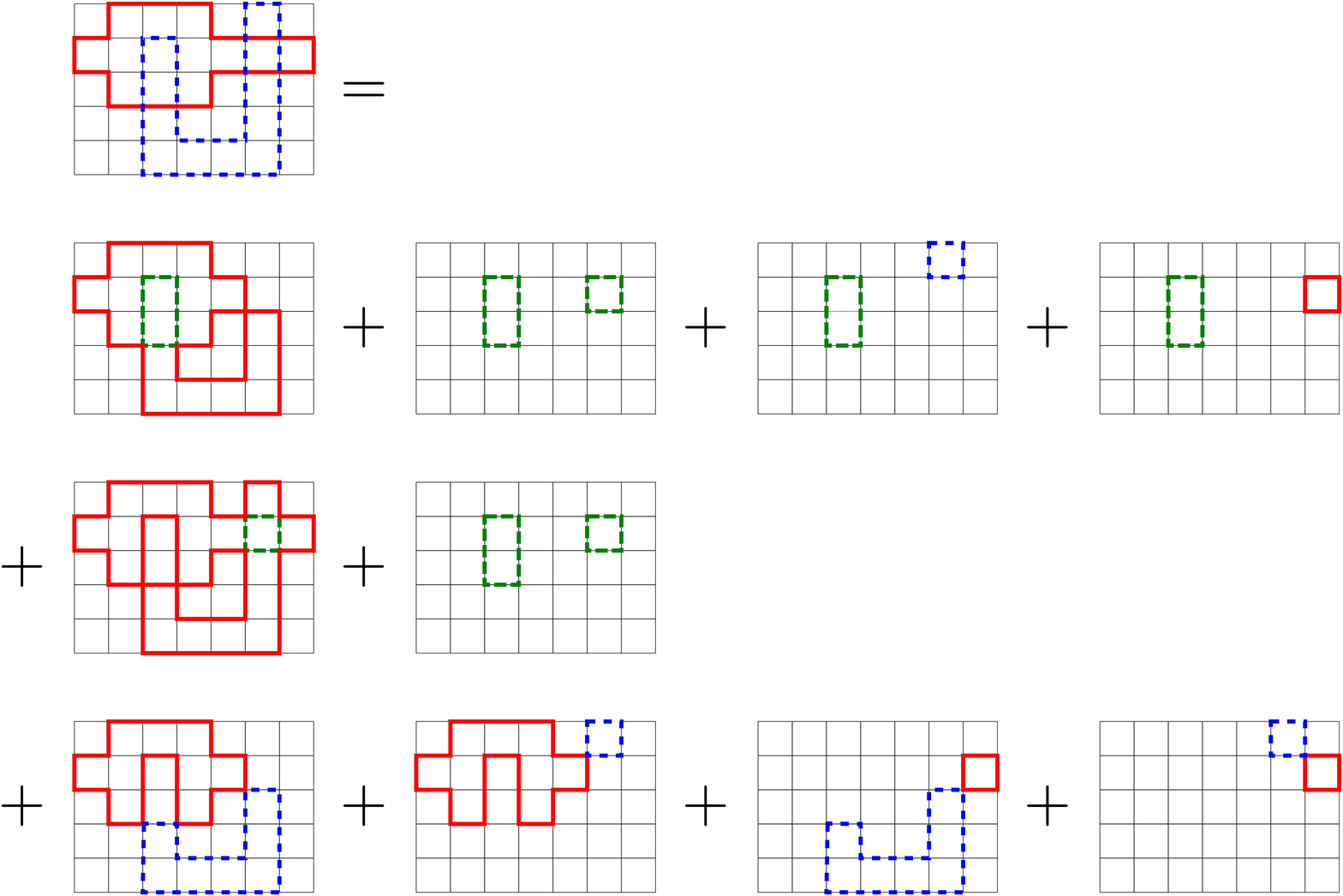}
\caption[Decomposition of intersecting cycles.]{Decomposition of the state associated to the cycle set in Fig.~\ref{fig:non_int} consisting of two intersecting cycles into a sum
of states associated to cycle sets with pairwise non-intersecting cycles. The combinations in the second and third line correspond to the first part of \eqref{eq:intersect}.
The fourth line corresponds to the second part of \eqref{eq:intersect}.}
\label{fig:proof}
\end{figure}

Let us now consider the general case of \(N\) particles. We may use the above result for \(N=2\) to prove the proposition by induction on the number of particles. We assume
that the statement holds for \(N-1\) cycles. A state associated to a cycle set may then be written as
\begin{equation}
 \left| \Phi(C) \right\rangle = \left( \sum_i \prod_{j=1}^{N-1} \sum_{f \in F(d_{ij})} b^\dagger_f \right) \sum_{f \in F(c)} b^\dagger_f  \left| 0 \right\rangle,
\end{equation}
with \(d_{ij} \circledcirc d_{ik} \) for \(j \neq k\). The \(N\)-th cycle is called \(c\). Let us examine an individual summand \(\left| \varphi_i(C) \right\rangle\) of
\(\left| \Phi(C) \right\rangle\).
\begin{equation}
 \left| \varphi_i(C) \right\rangle = \prod_{j=1}^{N-1} \sum_{f \in F(d_{ij})} b^\dagger_f \sum_{f \in F(c)} b^\dagger_f  \left| 0 \right\rangle.
\end{equation}
There is always a \(d_{ij}\) with \( F(d_{ik}) \not\subset F(d_{ij}) \) for all \(k \neq j\). Without loss of generality, let this \(d_{ij}\) be \(d_{iN-1}\). Since we know that
the statement holds for \(N=2\), we may construct \(\left| \varphi_i(C) \right\rangle\) as follows:
\begin{equation}
 \left| \varphi_i(C) \right\rangle =  \sum_m \left( \prod_{j=1}^{N-2} \sum_{f \in F(d_{ij})} b^\dagger_f 
\sum_{f \in F(e^2_{m2})} b^\dagger_f \sum_{f \in F(e^2_{m1})} b^\dagger_f \right) \left| 0 \right\rangle.
\end{equation}
The cycles \(e^2_{m1}\) and \(e^2_{m2}\) arise from the decomposition of  \(d_{iN-1}\) and \(c\) into non-intersecting cycles. We choose \(e^2_{m1} = \beta_k\), if the \(m\)-th
summand is a combination of cycles as in the first part of \eqref{eq:intersect} and \(e^2_{m1} = \alpha_i\) with \(F(\alpha_i) \subset F(d_{iN-1})\), if we have a combination
 as in the second part of \eqref{eq:intersect}. We then have \( e^2_{m1} \circledcirc e^2_{m2} \), \( e^2_{m1} \circledcirc d_{ij} \) for all \(j \in \{1, \ldots, N-2 \}\) and still
\( d_{ij} \circledcirc d_{ik} \) for \(j \neq k\). Hence, \( e^2_{m2} \) intersects at most \(N-2\) cycles. The construction also assures \( F(d_{ij}) \not\subset F(e^2_{m1}) \) for
all \(j\) and \( F(e^2_{m2}) \not\subset F(e^2_{m1}) \). This means that a cycle consisting of any edges of \( e^2_{m2} \) or \(d_{ij}\) cannot intersect \( e^2_{m1} \). If
we again examine an individual summand of \(\left| \varphi_i(C) \right\rangle\), we may therefore disregard \( e^2_{m1} \) and follow the same procedure as above with the \(N-2\)
remaining \(d_{ij}\) and \( e^2_{m2} \) in the role of \(c\). The result is a sum of states consisting of \(N-1\) pairwise non-intersecting cycles \( d_{ij},\ j\in \{1,\ldots, N-3 \}\),
\( e^3_{n2} \) and \( e^3_{n1} = e^2_{m1} \) and a cycle \( e^3_{n3} \) that intersects at most \( N-3 \) of the other cycles. If we iterate the process, after \(N-1\) steps we obtain
\begin{equation}
 \left| \Phi(C) \right\rangle = \sum_i \prod_{j=1}^N \sum_{f \in F(e^N_{ij})} b_f^\dagger \left| 0 \right\rangle,
\end{equation}
with \( e^N_{ij} \circledcirc e^N_{ik} \) for \(j \neq k\). This finally proves proposition \ref{prop:int}.
\end{proof}

Let us now proceed with the third step of our proof. We have to show that a state \( \left| \Phi (D^{(n)}) \right\rangle \) associated to a cycle set with pairwise non-intersecting
cycles may be written as a sum of states of \( \mathcal{B}_N \). We now give a procedure to construct this sum. Let us first recall
\begin{equation}
\label{eq:third_step}
 \left| \Phi (D^{(n)}) \right\rangle = \prod_{j=1}^N \sum_{f \in F(d_j)} b_f^\dagger \left| 0 \right\rangle.
\end{equation}
The product in \eqref{eq:third_step} shall be ordered in such a way that \( F(d_k) \not\subset F(d_j) \) for all \( j < k\). This is always possible. For any \( j < k \), we obtain either
\( F(d_j) \subset F(d_k) \) or \( f \cap g = \emptyset\), if \( f \in F(d_j) \) and \(g \in F(d_k) \) as a consequence of this ordering. If we expand the product in \eqref{eq:third_step},
we obtain all possible combinations of \( \prod_{j=1}^N b_{f_j}^\dagger \left| 0 \right\rangle \) with \( f_j \in F(d_j) \). Let us now construct a basis state
\( \left| \Phi (C_N^{(u)}) \right\rangle \in \mathcal{B}_N \) that contains such a combination.
\begin{equation}
 \left| \Phi (C_N^{(u)}) \right\rangle = \prod_{j=1}^N \sum_{f \in F(c_j)} b_f^\dagger \left| 0 \right\rangle.
\end{equation}
We choose \(c_1\) such that \(F(c_1)= \{f_1\}\). If \(f_1 \cap f_2 = \emptyset\), we set \( F(c_2)  = \left\{ f_2 \right\} \). If not, we choose
\( F(c_2)  = \{ f_1 \} \cup \{ f \ | \ f  \text{ adjacent to } f_1 \} \). Since \(F(d_1) \subset F(d_2) \) in this case, these faces are all
contained in \( F(d_2) \). This procedure is easily continued up to \( f_N \). If \( f_j \cap g = \emptyset \) for all \(g \in \bigcup_{k=1}^{j-1} F(c_k) \), we
choose \( F(c_j) \) to be \( \{f_j\} \). Otherwise, \( c_j \) has to be a cycle around possibly several \( c_k \),
\( k<j \). However, the ordering of the \( d_j \) in the first place assures that the faces which then have to be enclosed by \( c_j \) are as well contained in
\( F(d_j) \). In this way, all combinations of \( b^\dagger_{f_j} \) can be constructed and the state \( \left| \Phi (D^{(n)}) \right\rangle \) may be written as a sum of states
\( \left| \Phi (C_N^{(u)}) \right\rangle \in \mathcal{B}_N \).
This finally proves the theorem.

\section{Conclusions and Outlook}

We have given a basis for the many-particle ground states of the positive hopping bosonic Hubbard model on line graphs of finite 2-connected planar bipartite graphs at sufficiently
low filling factors. Each element of that basis consists of localized non-overlapping 
single-particle ground states which may be placed next to or wrapped around each other. A boson in these states is localized on a cycle of the line graph.
The construction works up to a critical filling factor at which the localized states are close-packed.
The different arrangements of the localized states give rise to a large ground state degeneracy especially at fillings below the close-packed case. 

It would be interesting to know
the exact degeneracy of the ground states. This would allow to determine the zero-temperature entropy per lattice site of the system 
which is finite for fillings below the critical filling factor. However, the computation of the
degeneracy has proven to be a difficult task and we are not able to answer this question here. 
For a planar periodic lattice, the states formed by placing localized states on the lattice so
that they do not overlap may be counted rather easily. Let \emph{e.g.} \(G\) be the honeycomb lattice.
The problem of placing localized states on the honeycomb lattice so
that they do not overlap can be mapped onto the exactly solvable hard-hexagon model 
\cite{PhysRevB.70.104415,PhysRevB.70.100403,PTPS.160.361}. 
The crucial problem is the counting of the remaining states in which cycles are enclosing each other. 
Therefore, the entropy of the exactly solvable hard-hexagon model is a lower bound to the entropy of the
bosonic Hubbard model on the kagome lattice. For general planar bipartite \(G\) the problem of counting
all sets of non-overlapping elementary cycles is difficult, so that for the general case already this
lower bound cannot easily be calculated.

Another intriguing problem is the generalization of our model to three dimensions. In this case, the construction of the single-particle basis we used in this work is not possible.
The reason is that we make use of Euler's Theorem which guarantees a one-to-one correspondence between the faces of the original graph \(G\) and the single particle ground states on
the line graph \(L(G)\). For non-planar \(G\) we loose this correspondence.

However, it is still possible to construct a basis for the single particle ground states using states localized on
elementary cycles. The dimension of the space of single-particle ground states is still \( D = |E(G)|-|V(G)|+1 \) \cite{0305-4470-24-2-005}. Hence, if one finds any \(D\) linearly independent
states localized on cycles of \( G \), they will form a basis. This can be easily done for \emph{e.g.} the cubic lattice. It is likely that all many-particle ground states can still be written
as linear combinations of states associated to cycle sets which has to be proved. Moreover, it should be possible to give a more general definition of contraction and uncontractibility of
cycle sets, if one allows deformations of cycles only by adding/subtracting cycles that correspond to the single-particle basis states.

At the critical filling factor (\emph{e.g.} 1/9 for the kagome lattice) the ground states of the system are close-packed. \emph{i.e.} Wigner crystals. The system becomes a Mott insulator.
A third open question is what happens when a small perturbation is added to the system. It is clear that a small perturbation to the single particle Hamiltonian which lifts the degeneracy
of the lowest band will lift ground states degeneracy of the multi-particle ground states as well. If one keeps translational invariance (no disorder), the single particle ground states will
be Bloch states which can be build by linear combinations of the localized states we used. The question then is whether there is still a Mott transition in this system.

\section*{References}

\providecommand{\newblock}{}

\bibliographystyle{unsrt.bst}

\end{document}